\pdfoutput=1
\documentclass[a4paper,UKenglish,numberwithinsect,cleveref, autoref, thm-restate]{lipics-v2019}
\usepackage{algorithm}
\usepackage[noend]{algpseudocode}

\title{Generalized Sorting with Predictions} 

\author{Pinyan Lu}{Shanghai University of Finance and Economics, China \and Huawei TCS Lab, China}{lu.pinyan@mail.shufe.edu.cn}{}{}

\author{Xuandi Ren}{Peking University, China}{renxuandi@pku.edu.cn}{}{}

\author{Enze Sun}{Shanghai Jiao Tong University, China}{sun\_en\_ze@sjtu.edu.cn}{}{}

\author{Yubo Zhang}{Peking University, China}{zhangyubo18@pku.edu.cn}{}{}

\authorrunning{P.\,Lu, X.\,Ren, E.\,Sun, and Y.\,Zhang}

\Copyright{John Q. Public and Joan R. Public} 



\category{} 

\relatedversion{} 

\supplement{}

\acknowledgements{}

\nolinenumbers 

\hideLIPIcs  

\EventEditors{John Q. Open and Joan R. Access}
\EventNoEds{2}
\EventLongTitle{42nd Conference on Very Important Topics (CVIT 2016)}
\EventShortTitle{CVIT 2016}
\EventAcronym{CVIT}
\EventYear{2016}
\EventDate{December 24--27, 2016}
\EventLocation{Little Whinging, United Kingdom}
\EventLogo{}
\SeriesVolume{42}
\ArticleNo{23}

\begin{document}

\maketitle

\begin{abstract}
Generalized sorting problem, also known as sorting with forbidden comparisons, was first introduced by Huang et al. ~\cite{huang2011algorithms} together with a randomized algorithm which requires $\tilde O(n^{3/2})$ probes. We study this problem with additional predictions for all pairs of allowed comparisons as input. We propose a randomized algorithm which uses $O(n \log n+w)$ probes with high probability and a deterministic algorithm which uses $O(nw)$ probes, where $w$ is the number of mistakes made by prediction.

\end{abstract}

\section{Introduction}
\subsection{Generalized sorting}
Sorting is arguably the most basic computational task, which is also widely used as an important component of many other algorithms. Pair-wise comparison is the core of most sorting algorithms. In the standard model, it is assumed that we can make comparison between all pairs of elements as we want and they are all of the same cost. However, this may not be the case in many applications. There might be some constrains which forbid us to compare some pairs of elements, or the costs for comparing different pairs of elements may be different. The non-uniform cost sorting model is studied in~\cite{charikar2002query,gupta2001sorting,khanna2003selection}. A special case called matching nuts and bolts problem is studied in \cite{alon1994matching,komlos1998matching}.

In this paper, we study the model introduced by Huang, Kannan and Khanna \cite{huang2011algorithms}, which is known as generalized sorting problem or sorting with forbidden pairs. In this model, only a subset of the comparisons are allowed, and each allowed comparison is of the same cost. We can view the elements as the vertices of a graph, each undirected edge in the graph represents a pair of elements which is allowed to be compared. A comparison of two elements $a,b$ leads to the exposure of the direction of edge $(a,b)$. It is guaranteed that the hidden directed graph is acyclic, and it contains a Hamiltonian path which represents the total order of all elements. Our goal is to adaptively probe edges in $E$ to find out the  Hamiltonian path.

For generalized sorting problem, the performance of an algorithm is measured by the number of edges it probes. In standard sorting problem where $G$ is a complete graph, the minimum number of probes required is $\Theta(n \log n)$. For general graphs, one may need more probes. Huang et al. \cite{huang2011algorithms} have proved an upper bound of $\widetilde O(n^{1.5})$ on the number of probes by giving a randomized algorithm. When the graph is dense and the number of edges is as large as $\binom{n}{2}-q$, \cite{Banerjee2016SortingUF} proposes a deterministic algorithm which makes $O((n+q) \log n)$ probes together with a randomized algorithm which makes $\widetilde{O}\left(n^{2} / \sqrt{q+n}+n \sqrt{q}\right)$ probes with high probability. Most part of the generalized sorting problem is still open.

\subsection{Algorithms with predictions}
Recently, there is an interesting line of research called algorithm design with predictions \cite{mitzenmacher2020algorithms}, which is motivated by the observation that by making use of predictions provided by machine learning, one may be able to design a more effective algorithm. Normally,  the better the prediction, the better the performance. In this framework, we aim for algorithms which have near optimal performance when the predictions are good, and no worse than the prediction-less case when the predictions have large errors. The above two targets in algorithm design with predictions are called consistency and robustness respectively.

Take the classic binary search algorithm as an example, which can find the position of an existing element in a sorted list in $O(\log n)$ comparisons. It starts by querying the median of the list. However, if a machine learning algorithm can roughly estimate the position of the given element, it may not be always a good idea to start from the middle. Based on this idea, one designed an algorithm with query complexity of $O(\log w)$, where $w$ is the distance between the true position of the element and the estimated one, which measures the accuracy of the prediction. This algorithm can be much better than $O(\log n)$ when $w$ is much smaller than $n$, and no worse than the prediction-less binary search even if the estimation is terribly wrong because $w$ is at most $n$.

Algorithms with predictions are studied for caching~\cite{lykouris2018competitive,rohatgi2020near}, ski-rental, online scheduling~\cite{purohit2018improving} and other problems. See the nice survey by Mitzenmacher and Vassilvitskii \cite{mitzenmacher2020algorithms}.

\subsection{Our results}
In this paper, we initiate the study of generalized sorting with predictions. The model is very natural, besides the undirected graph $G=(V, E)$, we are also given an orientation of $G$ as input, which are predictions of the hidden direction of the edges. The number of mis-predicted edges is denoted by $w$.  With the help of predictions, we hope to improve the bound of $\tilde O(n^{1.5})$ when $w$ is small. 

In section 3, we propose a randomized algorithm for the generalized sorting problem and prove that it probes at most $O(n \log n+ w)$ edges with high probability. The description of the algorithm is simple while the analysis is quite subtle and involved.

\begin{restatable}[An $O(n \log n+w)$ randomized algorithm]{theorem}{mainthm}
\label{thm:mainthm1}
There is a polynomial time randomized algorithm which solves generalized sorting problem in $O(n \log n+w)$ probes with high probability.
\end{restatable}

In section 4, we also propose a deterministic algorithm using $O(nw)$ probes, in order to show that when $w$ is as small as a constant, the generalized sorting with prediction problem can be solved using only linear probes.

\begin{restatable}[An $O(nw)$ deterministic algorithm]{theorem}{secthm}
\label{thm:mainthm2}
There is a polynomial time deterministic algorithm which solves generalized sorting problem in $O(nw)$ probes.
\end{restatable}

Note that in the query complexity model, if we have two algorithms $\mathcal A$ and $\mathcal B$ which use $O(f(n))$ and $O(g(n))$ queries respectively, we can simply merge them into an algorithm $\mathcal C$, which simulates $\mathcal A$ and $\mathcal B$, and make $\mathcal A$'s queries and $\mathcal B$'s queries alternately. Then $\mathcal C$ uses only $O(\min(f(n),g(n)))$ queries. Therefore by combining our algorithms with the one in \cite{huang2011algorithms}, both consistency and robustness can be achieved.

\section{Preliminaries}
The input of the generalized sorting with prediction problem is an undirected graph $G=(V,E)$ together with an orientation $\vec P$ of $E$. There is another orientation $\vec E$ of $E$ which represents the underlying total order and is unknown to us. The problem is formally stated as follows:

\begin{definition}[generalized sorting with prediction]
	An instance of generalized sorting with prediction problem can be represented as $(V,E,\vec E,\vec P)$, where
	\begin{itemize}
		\item $G=(V,E)$ is an undirected graph.
		\item $\vec P,\vec E$ are two orientations of $E$, i.e. $\forall (u,v) \in E$, exactly one of $(u,v) \in \vec P$ and $(v,u) \in \vec P$ holds, and exactly one of $(u,v) \in \vec E$ and $(v,u) \in \vec E$ holds.
		\item $\vec G=(V,\vec E)$ is the directed graph which represents the underlying total order, it is guaranteed that $\vec G$ is acyclic and there is a directed Hamiltonian path in it.
		\item $\vec G_P=(V,\vec P)$ is the predicted directed graph, there are no more guarantees about $\vec G_P$.
	\end{itemize}	
\end{definition}

An edge $(u,v)$ whose direction is different in $\vec{P}$ and $\vec{E}$ is called a \textit{mis-predicted edge}. An in-neighbor of $u$ in $\vec{G}_P$ which is not an in-neighbor of $u$ in $\vec{G}$ is called a \textit{wrong in-neighbor}.

We use $n=|V|$ to denote the number of vertices and $w=|\vec P \backslash \vec E|$ to denote the number of mis-predicted edges. The input and the required output of the problem are stated as follows:
\begin{itemize}
	\item Input: $(V,E,\vec P)$
	\item Output: $(v_1,v_2,...,v_{n})$ s.t. $\forall 1 \le i < n, (v_i,v_{i+1}) \in \vec E$, which represents the directed Hamiltonian path in $\vec G$.
\end{itemize}
	
Note that $\vec E$ is not given as input, but is fixed at the very beginning and does not change when the algorithm is executing.

An algorithm can adaptively probe an edge and know its direction in $\vec E$. The performance of an algorithm is measured by the number of edges it probes, which may be in terms of $n$ and $w$.

Recall that $\vec G=(V,\vec E)$ is acyclic, so any subset $\vec E' \subseteq \vec E$ naturally defines a partial order of $V$. Sometimes we focus on a vertex set $V'$, and only consider the partial order in the induced subgraph $\vec G[V']$:

\begin{definition}[partial order $<_{\vec E'}$ and $<_{V'}$]
~
	\begin{itemize}
		\item $\vec E'\subseteq \vec E$ defines a partial order of $V$ on graph $(V,\vec E')$, which is referred to as $<_{\vec E'}$. For $a,b \in V$, $(a<_{\vec E'}b)$ iff there is a directed path from $a$ to $b$ which only consists of edges in $\vec E'$.
		\item $V' \subseteq V$ defines a partial order of $V'$ on the induced subgraph $\vec G[V']$, which is referred to as $<_{V'}$. For $a,b \in V'$, $(a<_{V'} b)$ iff there is a directed path from $a$ to $b$ which only consists of edges in $\vec E$ and only passes vertices in $V'$.
	\end{itemize}
\end{definition}

\begin{definition}[vertex sets $\mathcal N_{in}(G_{\vec P},u),S_u,T_u$]~

	Denote by $\mathcal N_{in}(G_{\vec P},u)$ the set of all in-neighbors of $u$ in the prediction graph, i.e. $\mathcal N_{in}(G_{\vec P},u)=\{v|(v,u) \in \vec P\}$.
	
	Denote by $S_u$ the set of real in-neighbors of $u$ among $\mathcal N_{in}(G_{\vec P},u)$, i.e. $S_u=\{v|(v,u) \in \vec P \cap \vec E\}$.
	
	Denote by $T_u$ (with respect to a specific moment) the set of in-neighbors of $u$, which are not known to be wrong at that moment, i.e. the corresponding edges are either correct or unprobed. $T_u=\{v| (v,u) \in \vec P \land (u,v) \notin \vec Q\}$ where $\vec Q$ (also with respect to a specific moment) is the set of probed directed edges up to that moment.
\end{definition}

Note by definition it always holds $S_u \subseteq T_u \subseteq \mathcal N_{in}(G_{\vec P},u)$. $S_u$ and $\mathcal N_{in}(G_{\vec P},u)$ are fixed while $T_u$ may change over time. Initially there are no probed edges, so $T_u=\mathcal N_{in}(G_{\vec P},u)$. As the algorithm proceeds, some mis-predicted edges between $\mathcal N_{in}(G_{\vec P},u)$ and $u$ are found, the corresponding wrong in-neighbors no longer belong to $T_u$ and $T_u$ will finally shrink to $S_u$.

\section{An algorithm using $O(n \log n+w)$ probes} 

\subsection{Description}
The algorithm maintains a set of vertices $A$ satisfying $\forall u \in A$, direction of edges between $\mathcal N_{in}(\vec G_P,u)$ and $u$ are all known to us (either probed or can be deduced from other probed edges). Notice that the direction of edges in the induced subgraph $G[A]$ must be all known. When $A=V$, the direction of all edges are known and we can easily find the desired Hamiltonian path. 

We initialize $A$ as $\emptyset$, then iteratively add `ideal vertices' to $A$, which are defined as follows:

\begin{definition}[ideal vertex]
A vertex $u\in V$ is called an \textit{ideal vertex} if both of the following conditions are satisfied:
	\begin{enumerate}
		\item $T_u \subseteq A$.
		\item The partial order $<_A$ restricted to $T_u$ is a total order.
	\end{enumerate}
\end{definition}

Before adding a vertex $u$ to $A$, we need to determine the direction of edges between $T_u$ and $u$ (those between $\mathcal N_{in}(\vec G_P,u) \backslash T_u$ and $u$ have been already known to be mis-predicted). For an ideal vertex $u$, this can be done by using a straightforward strategy: repeatedly probe the edge $(t,u)$, where $t$ is the largest vertex in $T_u$ with respect to $<_A$. If the direction of this edge is correct, i.e. $t<_{\vec E} u$, we can conclude that the direction of all edges between $T_u$ and $u$ are correct by transitivity. We can end this phase and add $u$ to $A$. Otherwise $(t,u)$ is a mis-predicted edge, $t$ is removed from $T_u$ and we move on to probe the edge between the new largest vertex in $T_u$ and $u$, and so on.

If there is an ideal vertex, we are in an ideal case: by probing only one edge, we either learn the direction of all edges between $T_u$ and $u$ and add $u$ into $A$, or find a mis-predicted edge. Notice that each vertex is added to $A$ once, and the wrong probes are charged to the $w$ term of complexity. Therefore we can add all vertices to $A$ in only $O(n+w)$ probes, assuming there is \textbf{always} an ideal vertex in each step. 

However, the assumption does not always hold due to the existence of mis-predicted edges, i.e. there may be a time when there is no ideal vertex. We have to relax the conditions to define a new type of vertex to help, which always exists:

\begin{definition}[active vertex]
A vertex $u\in V$ is called an \textit{active vertex} if both of the following conditions are satisfied:
\begin{enumerate}
	\item $S_u\subseteq A$.
	\item The partial order $<_A$ restricted to $S_u$ is a total order.
\end{enumerate} 
\end{definition}
 
\begin{lemma}\label{actexist}
 	There is at least one active vertex in $V \backslash A$ if $A \ne V$.
\end{lemma}
\begin{proof}
	Suppose the Hamiltonian path in $\vec G$ is $(v_1,...,v_n)$. Let $k$ be the smallest index s.t. $v_k \notin A$, then $v_{k}$ satisfies
 	\begin{enumerate}
 		\item $S_{v_{k}}\subseteq\{v_1,...,v_{k-1}\} \subseteq A$.
		\item $<_A$ restricted to $S_{v_k}$ is a total order.
	\end{enumerate}
	Therefore $v_{k}$ is active at the moment.
\end{proof}

An ideal vertex is always active since $S_u \subseteq T_u$ holds, but there may be some wrong in-neighbors not identified yet (the vertices in $T_u\backslash S_u$) to prevent an active vertex from being ideal. By cleverly identify the wrong in-neighbors and remove them from $T_u$, an active vertex $u$ would become an ideal one and we can use the above strategy again.

As $S_u$ is invisible to us, we know neither which vertices are active, nor which in-neighbors of a vertex are wrong, so we turn to focus on the in-neighbors of $u$ which prevent it from being ideal:

\begin{enumerate}
	\item $v \in T_u$ s.t. $v \notin A$.
	\item $v_1,v_2 \in T_u$ s.t. $(v_1,v_2 \in A) \land (v_1 \not <_A v_2) \land (v_2 \not <_A v_1)$.
\end{enumerate}

Now consider an active vertex $u$. If such $v$ in case 1 exists, then $(v,u)$ must be a mis-predicted edge. If such $v_1,v_2$ in case 2 exists, there is at least one mis-predicted edge in $(v_1,u),(v_2,u)$. By probing $(v,u)$ or both $(v_1,u),(v_2,u)$ repeatedly, we can keep removing its wrong in-neighbors from $T_u$ and finally make $u$ ideal.

For an inactive vertex $u$, the direction of $(v,u)$ or both of $(v_1,u),(v_2,u)$ may be correct, but that tells us $u$ is not active hence is not the vertex we are looking for. In this case the vertex $v$ or the pair of vertices $(v_1,v_2)$ is called a \textit{certificate} for $u$, which proves that $u$ is currently not active. 

\begin{definition}[certificate]
	For a vertex $u \in V$,
	\begin{enumerate}
 		\item A type-1 certificate is a vertex $v\in S_u$ s.t. $v \notin A$.
 		\item A type-2 certificate is a pair of different vertices $v_1,v_2 \in S_u$ s.t. $(v_1,v_2 \in A)\land (v_1 \not<_A v_2) \land (v_2 \not<_A v_1)$.
	\end{enumerate}
\end{definition}

Once a certificate for $u$ is found, we turn to check the activeness of other vertices and do not need to probe any other incoming edges for $u$ until the next vertex is added to $A$. 
For a fixed set $A$, both activeness of vertices and validity of certificates are determined and do not change when new probes are made. Only when $A$ is extended and the current certificate of $u$ is no longer valid do we need to look for a new certificate for $u$. A type-1 certificate $v$ becomes invalid when $v$ is added into $A$, while a type-2 certificate $(v_1,v_2)$ becomes invalid when $(v_1 <_A v_2) \lor (v_2 <_A v_1)$ happens as $A$ expands.

In the worst case, one may need to update the certificates again and again and thus probe too many edges. By checking the validity of certificates in a random order, the worst case is avoided with high probability. We prove that our algorithm uses only $O(n \log n+w)$ probes with high probability in the next subsection, where the term $n \log n$ comes from the probes used in re-searching for valid certificates.

Our algorithm works by repeatedly choose a vertex $u$ which does not have a valid certificate. 
    \begin{enumerate}
        \item If it is an ideal vertex, we use the strategy mentioned above to determine the direction of edges between $T_u$ and $u$, then add $u$ to $A$. 
        \item Otherwise there must be a vertex $v \in T_u$ s.t. $v \notin A$, or a pair of vertices $v_1,v_2 \in T_u$ s.t. $(v_1,v_2 \in A) \land (v_1\not<_A v_2) \land (v_2 \not<_A v_1)$. We randomly choose such a vertex $v$ or such a pair of vertices $(v_1,v_2)$ and probe the edge(s) between $u$ and them. Then either at least one mis-predicted edge is found, or a valid certificate for $u$ is found. 
    \end{enumerate}

 Since there is always an active vertex $u$, after finding some mis-predicted edges and removing the corresponding wrong in-neighbors from $T_u$, $u$ must become an ideal vertex, that is how the algorithm makes progress.

Here is the pseudo code of the algorithm: 
 
\begin{algorithm}
 \caption{A randomized algorithm using $O(n \log n+w)$ probes}
 \begin{algorithmic}[1]
 \State set $A:=\emptyset$
 	\While{$A \ne V$}
 	\State pick $u \in V \backslash A$ s.t. $u$ does not have a valid certificate (if there are multiple ones, pick one with the smallest index)
 	\If {$T_u \not \subseteq A$}
 		\State randomly pick $v \in T_u \backslash A$
 		\State probe $(v,u)$
 	\ElsIf {$\exists$ different $v_1,v_2 \in T_u$ s.t. $(v_1\not<_A v_2) \land (v_2 \not<_A v_1)$}
    	\State randomly select such a pair $v_1,v_2$
	    \State probe $(v_1,u),(v_2,u)$
	\Else
    	\State let $t$ be the largest vertex in $T_u$ w.r.t.  $<_A$
 		\State probe $(t,u)$
    	\If{the direction of $(t,u)$ is correct, i.e. $t<_{\vec E} u$}
    		\State add $u$ to $A$
    	\EndIf
    \EndIf
 
  \EndWhile
\end{algorithmic}
\end{algorithm}

\subsection{Analysis}~

We first prove that the algorithm can always proceed, and always terminates.

\begin{lemma}\label{corr}
	The algorithm can always proceed, and will terminate in finitely many steps. 
\end{lemma}
\begin{proof}
We know from Lemma \ref{actexist} that there is always at least one active vertex outside $A$ if $A \ne V$. Since active vertices have no valid certificates, in each execution of line 3, there is always at least one vertex which meets our requirements.
 	
In each execution of the `while' loop, exactly one of the following happens:
\begin{enumerate}
	\item We find a certificate for $u$ which doesn't have a valid one previously.
	\item We find a mis-predicted edge.
	\item We add a vertex $u$ into $A$.
\end{enumerate}
 	
When case 1 happens, we find a new certificate for $u$. Only when this certificate becomes invalid do we need to look for a new one for $u$. An invalid certificate will never become valid again since $A$ is always enlarging. Therefore this case happens for finitely many times.

Case 2 happens for finitely many times because once we find a mis-predicted edge, we remove a vertex from $T_u$, the size of which is initially finite and non-negative all the time.
 	
Case 3 also happens for finitely many times because each vertex is added into $A$ only once.
\end{proof}
 
Now we proceed to analyze the total number of probes made by the algorithm. First we introduce some important lemmas:
 
\begin{lemma}\label{fix1}
	All vertices are added into $A$ in a fixed order regardless	of the randomness.
\end{lemma}
\begin{proof}
	Recall that an ideal vertex is always active, so we only add active vertices to $A$. Whether a vertex $u$ is active or not only depends on the current $A$. Therefore $\forall 1 \le i \le n$, when $|A|=i-1$, the $i$-th vertex added into $A$ is always the one with the smallest index among all active vertices in $V \backslash A$ regardless of randomness.
\end{proof}

\begin{corollary}\label{fix2}
    Let $C_{(a,b)}(a,b\in V)$ denotes the event that the vertex pair $(a,b)$ becomes comparable in $<_A$, i.e. $(a,b\in A) \land ((a<_A b) \lor (b<_A a))$. As $A$ expands, all $\binom{n}{2}$ such events happen in a fixed order regardless of the randomness (breaking ties in an arbitrarily fixed manner).
\end{corollary}
{\bf Remark.} Lemma \ref{fix1} and Corollary \ref{fix2} use the fact that we determine the order among vertices in $T_u$ not according to all edges probed till now, but only according to the edges in the induced subgraph $G[A]$. It seems more efficient if we use all the information instead of the restricted portion, but it will create subtle correlations and we do not know how to analyze. The above fixed-order properties are crucial in the proof of our main theorem.
 
\begin{lemma}\label{rand}
	For a random permutation $\{Y_1,...,Y_n\}$ of $\{1,...,n\}$, the number of elements $Y_i$ s.t. $Y_i=\max_{j \le i} Y_j$ does not exceed $6 \ln n+6$ w.p. at least $1-\frac{1}{2n^2}$.
\end{lemma}
\begin{proof}
	A random permutation can be built in such a way: 
	\begin{itemize}
 		\item $\forall 1 \le i \le n$, randomly and independently pick $Z_i$ from $\{1,...,i\}$.
	 	\item take the unique permutation $\{Y_1,...,Y_n\}$ s.t. $\forall i, Y_i$ is the $Z_i$-th largest element in $\{Y_1,...,Y_i\}$.
	\end{itemize}
 
 It's easy to see a random permutation built in this way is uniformly distributed.
 
 Let $\{X_1,...,X_n\}$ be a sequence of 0-1 random variables  indicating whether $Z_i=i$. The number mentioned in the lemma is just $X=\sum_{i=1}^n X_i$ since $Z_i=i \Leftrightarrow Y_i=\max_{j \le i} Y_j$. We have
 $$\Pr[X_i=1]=1-\Pr[X_i=0]=\frac{1}{i}$$
 
 Note that $\mu=\mathbb E[X]=H_n=\sum_{i=1}^{n}\frac{1}{i} \approx \ln n+\gamma$  as $n \rightarrow \infty$ where $\gamma$ is the euler constant, and
  $$\ln n+\gamma < H_{n} \le \ln n+1, \forall n \ge 1$$
 
 Plugging $\epsilon=5$ into a Chernoff bound:
  $\Pr[X>(1+\epsilon) \mu] \le \exp\left(-\frac{\mu \epsilon^2}{2+\epsilon}\right)$, we have
  $$\begin{aligned}
	\Pr[X>6(\ln n+1)]& \le \exp\left(-\frac{25}{7}(\ln n+\gamma)\right)\\
	& \le \frac{1}{2n^2}\end{aligned}$$
      \end{proof}

\mainthm*

 \begin{proof}
 The correctness of the algorithm directly follows from  Lemma \ref{corr}. We only need to bound the number of probes it uses.
 
 Follow the proof of Lemma \ref{corr} we know in each execution of the `while' loop, exactly one of the three cases happens:
 \begin{itemize}
 	\item We find a certificate for $u$ which doesn't have a valid one previously.
 	\item We find a mis-predicted edge.
 	\item We add a vertex $u$ into $A$.
 \end{itemize}
 	
The algorithm makes no more than two probes in each loop, so it's sufficient to bound the number of occurrences of each cases.
 
Case 2 and case 3 happen for at most $n+w$ times in all, because the number of mis-predicted edges is at most $w$ and each vertex is added into $A$ exactly once.
 
We now focus on case 1. Once a valid certificate for $u$ is found, we won't make any further probes for $u$ until its current certificate becomes invalid. According to Lemma \ref{fix1}, all vertices are added into $A$ in a fixed order, which means all possible type-1 certificates for $u$ (the set of which is $S_u$ initially) become invalid in a fixed order. 

Each time we find a uniformly random type-1 certificate for $u$ among its currently valid ones. In the analysis it can be equivalently viewed as that for each $u$, a random permutation $\{P_{1}^{(u)},...,P_{|S_u|}^{(u)}\}$ of $S_u$ is chosen and fixed at first, and all type-1 certificates used in the process are identified according to this permutation: when a new valid type-1 certificate is found, let it be $P_{i}^{(u)}$, where $i$ is the smallest index s.t. $P_{i}^{(u)}$ is currently valid as a type-1 certificate for $u$.

Let $\{Y_{1}^{(u)},...,Y_{|S_u|}^{(u)}\}$ represents the fixed order of $\{P_{1}^{(u)},...,P_{|S_u|}^{(u)}\}$ to become invalid, i.e. $P_{i}^{(u)}$ is the $Y_{i}^{(u)}$-th earliest to become invalid. $\{Y_{1}^{(u)},...,Y_{|S_u|}^{(u)}\}$ is a uniformly random permutation as well as $\{P_{1}^{(u)},...,P_{|S_u|}^{(u)}\}$, and the total number of valid type-1 certificates found for $u$ equals to the number of $Y_i^{(u)}$ s.t. $Y_i^{(u)}=\max_{j \le i} Y_j^{(u)}$.
 
From Lemma \ref{rand} we know w.p. at least $1-\frac{1}{2n^2}$, this number does not exceed $6 \ln n+6$. Take union bound over all $u$, w.p. at least $1-\frac{1}{2n}$, no vertex uses more than $(6 \ln n+6)$ type-1 certificates, hence total number of valid type-1 certificates the algorithm finds does not exceed $6 n \ln n+6n$.
 
The above analysis is exactly the same for type-2 certificates, since according to Corollary \ref{fix2}, the possible type-2 certificates for $u$ also become invalid in a fixed order. The only difference is that the number of valid type-2 certificates for each $u$ may be up to $n^2$, while it is at most $n$ for type-1 certificates. Again use Lemma \ref{rand} and take union bound, w.p. at least $1-\frac{1}{2n}$, the total number of valid type-2 certificates the algorithm finds does not exceed $12 n\ln n+6n$. 
 
 Combining all cases we can conclude that w.p. at least $1-\frac{1}{n}$, the algorithm uses no more than $O(n \log n+w)$ probes in total.
\end{proof}

\section{An algorithm using $O(nw)$ probes}

Here we briefly introduce a deterministic algorithm using $O(nw)$ probes, to show that the generalized sorting with prediction problem can be solved in only linear probes when $w$ is as small as a constant. 

The basic idea is to find a mis-predicted edge in $O(n)$ probes and correct it in the predicted graph. We use $\vec G_C=(V,\vec P_C)$ to denote the predicted graph after correction: $\vec P_C=\{(v,u)\in\vec{P}|(u,v)\notin \vec{Q}\}\cup \vec{Q}$ where $\vec Q\subseteq \vec E$ is the set of directed edges probed till now.

If there is a directed cycle in $\vec G_C$, there must be at least one mis-predicted edge on the cycle since the actual $\vec G$ is acyclic. We just probe all edges on a simple directed cycle, update $\vec G_C$ and loop again. 

If $\vec G_C$ is acyclic, consider running topological sort on it. If the direction of all edges in $G_C$ are correct, each time there should be exactly one vertex whose in-degree is 0. If not, i.e. there are two vertices $v_1,v_2$ with in-degree 0 at the same time, this can only happen when there are some mis-predicted edges either adjacent to $v_1,v_2$ or on the path produced by the topological sort before. We probe all such edges and loop again.

The pseudo code of the algorithm is stated as follows:


 \begin{algorithm}
 \caption{A deterministic algorithm using $O(nw)$ probes}
 \begin{algorithmic}[1]
 \While {there is no Hamiltonian path consisting of only probed edges}
  \If{there is a simple directed cycle in $\vec{G}_C=(V,\vec P_C)$}
  \State probe all the edges on the cycle
  \Else
  	\State run topological sort on $\vec G_C$ and stop when $\exists v_1,v_2$ both with in-degree 0
  	\State let $(a_1,...,a_k)$ be the (partial) topological order
  	\If {$k = n$ (i.e. no such $v_1,v_2$ found)}
		\State $\forall 1 \le i<k$, probe the edge $(a_i,a_{i+1})$
  	\Else
		\State $\forall 1 \le i<k$, probe the edge $(a_i,a_{i+1})$
  		\State probe all edges adjacent to $v_1,v_2$
  	\EndIf
  \EndIf
  \EndWhile
 \end{algorithmic}
  \end{algorithm}
  
\secthm*
  
  \begin{proof}
  In each execution of the `while' loop, exactly one of the following three cases happens, and we analyze them separately:
  \begin{enumerate}
    \item \textit{There is a simple directed cycle in $\vec{G}_C$.} Since the actual $\vec{G}$ is acyclic, at least one edge on the cycle in $\vec G_C$ is mis-predicted. By probing all edges on it we will find at least one mis-predicted edge. 
    \item \textit{$\vec{G}_C$ is acyclic and $k=n$ in line 7, i.e. the topological sort terminates normally.} By probing all edges between the adjacent vertices in the topological order, we either find a mis-predicted edge, or can know that the resulting path is just the desired Hamiltonian path. 
    \item \textit{$\vec G_C$ is acyclic and there are two different vertices $v_1,v_2$ with in-degree 0 during the topological sort.} 

    In this case, a mis-predicted edge must be found in line 10 or line 11. Prove it by contradiction, assume all edges probed in line 10 and line 11 are correct, if $k>0$, $(a_k,v_1)$ and $(a_k,v_2)$ must both lie in $\vec G_C$ since the topological sort stops just after handling $a_k$, so $\forall 1 \le i\le k, (a_i<_{\vec E} v_1) \land (a_i<_{\vec E} v_2)$ holds due to transitivity. W.l.o.g assume $v_1<_{\vec E}v_2$. Consider the directed path from $v_1$ to $v_2$ in the actual graph $\vec G$, let it be $(b_1,...,b_l)$  where $b_1=v_1$ and $b_l=v_2$. Note that $a_1<_{\vec E}...<_{\vec E} a_k<_{\vec E} b_1=v_1<_{\vec E} ... <_{\vec E} b_l=v_2$. Therefore the edge $(b_{l-1},b_l) \in \vec G_C$ and $b_{l-1} \notin \{a_1,...,a_k\}$, which contradicts the fact that the in-degree of $v_2$ is 0 at that moment.
   \end{enumerate}
   
  Therefore in each loop, we either find at least one mis-predicted edge in $\vec{G}_C$ or find the correct Hamiltonian path. It's obvious that the number of probes we make is $O(n)$ in each loop, so the total number of probes does not exceed $O(nw)$.
  \end{proof}

\bibliography{sorting}

\begin{thebibliography}{10}

\bibitem{alon1994matching}
Noga Alon, Manuel Blum, Amos Fiat, Sampath Kannan, Moni Naor, and Rafail
  Ostrovsky.
\newblock Matching nuts and bolts.
\newblock In {\em SODA}, pages 690--696, 1994.

\bibitem{Banerjee2016SortingUF}
I.~Banerjee and D.~Richards.
\newblock Sorting under forbidden comparisons.
\newblock In {\em SWAT}, 2016.

\bibitem{charikar2002query}
Moses Charikar, Ronald Fagin, Venkatesan Guruswami, Jon Kleinberg, Prabhakar
  Raghavan, and Amit Sahai.
\newblock Query strategies for priced information.
\newblock {\em Journal of Computer and System Sciences}, 64(4):785--819, 2002.

\bibitem{gupta2001sorting}
Anupam Gupta and Amit Kumar.
\newblock Sorting and selection with structured costs.
\newblock In {\em Proceedings 42nd IEEE Symposium on Foundations of Computer
  Science}, pages 416--425. IEEE, 2001.

\bibitem{huang2011algorithms}
Zhiyi Huang, Sampath Kannan, and Sanjeev Khanna.
\newblock Algorithms for the generalized sorting problem.
\newblock In {\em 2011 IEEE 52nd Annual Symposium on Foundations of Computer
  Science}, pages 738--747. IEEE, 2011.

\bibitem{khanna2003selection}
Sampath Kannan and Sanjeev Khanna.
\newblock Selection with monotone comparison costs.
\newblock In {\em Proceedings of the fourteenth annual ACM-SIAM symposium on
  Discrete algorithms}, page~10. SIAM, 2003.

\bibitem{komlos1998matching}
J{\'a}nos Koml{\'o}s, Yuan Ma, and Endre Szemer{\'e}di.
\newblock Matching nuts and bolts in {$O (n \log n)$} time.
\newblock {\em SIAM Journal on Discrete Mathematics}, 11(3):347--372, 1998.

\bibitem{lykouris2018competitive}
Thodoris Lykouris and Sergei Vassilvitskii.
\newblock Competitive caching with machine learned advice.
\newblock {\em arXiv preprint arXiv:1802.05399}, 2018.

\bibitem{mitzenmacher2020algorithms}
Michael Mitzenmacher and Sergei Vassilvitskii.
\newblock Algorithms with predictions.
\newblock {\em arXiv preprint arXiv:2006.09123}, 2020.

\bibitem{purohit2018improving}
Manish Purohit, Zoya Svitkina, and Ravi Kumar.
\newblock Improving online algorithms via ml predictions.
\newblock In {\em Advances in Neural Information Processing Systems}, pages
  9661--9670, 2018.

\bibitem{rohatgi2020near}
Dhruv Rohatgi.
\newblock Near-optimal bounds for online caching with machine learned advice.
\newblock In {\em Proceedings of the Fourteenth Annual ACM-SIAM Symposium on
  Discrete Algorithms}, pages 1834--1845. SIAM, 2020.

\end{thebibliography}

\end{document}